%% file: SOSA22.tex
\documentclass[11pt]{article}
\RequirePackage{algorithm,algorithmic,graphicx,amssymb,epsf,epic,url,complexity}
\RequirePackage{fullpage}
\usepackage{enumitem}
\usepackage{authblk}
\usepackage{mdframed}
\usepackage{microtype, complexity}
\usepackage[usenames,dvipsnames]{xcolor}
\usepackage{amsthm}
\usepackage{amsmath}
\usepackage{upref}
\usepackage{amsfonts}
\usepackage{graphicx}

\usepackage{latexsym}
\usepackage{amssymb}
\usepackage{amscd}
\usepackage{mdframed}
\usepackage[all]{xy}

\usepackage{xcolor}
\usepackage{amsthm}
\usepackage{framed}

\usepackage{mathrsfs}
\usepackage{amsfonts}
\usepackage{epsfig}
\usepackage{epstopdf}
\usepackage{tcolorbox}
\usepackage{tikz}
\usepackage[hang, small,labelfont=bf,up,textfont=it,up]{caption} % Custom captions under/above floats in tables or figures

\usepackage[backref=page, pagebackref=true]{hyperref}
\usepackage[capitalize]{cleveref} % cleveref after hyperref

\newcommand {\ignore} [1] {}
\def\eps{\varepsilon}

\colorlet{shadecolor}{gray!15}

\newtheorem{thm}{Theorem}

\newenvironment{theo}
  {\begin{shaded}\begin{thm}}
  {\end{thm}\end{shaded}}
  
\theoremstyle{definition}

\theoremstyle{remark}

\def\polylog{\operatorname{polylog}}

\input{latexmac}

\denseformat
\newtheorem{mdresult}[theorem]{Theorem}

\newtheorem{question}{Question}
\newtheorem{mdresult2}[question]{Question}
\newenvironment{Question}{\begin{mdframed}[backgroundcolor=lightgray!40,topline=false,rightline=false,leftline=false,bottomline=false,innertopmargin=2pt]\begin{mdresult2}}{\end{mdresult2}\end{mdframed}}

\begin{document}

\title{Maintaining an EDCS in General Graphs: \\Simpler, Density-Sensitive and with Worst-Case Time Bounds}
\author{Fabrizio Grandoni}
\affil{IDSIA, USI-SUPSI}
\author{Chris Schwiegelshohn} 
\affil{Aarhus University}
\author{Shay Solomon\thanks{Research was partially supported by Israel Science Foundation (ISF) grant 1991/19,
and by a grant from the United States-Israel Binational Science Foundation (BSF) and the United States National Science Foundation (NSF).
}}
%\affil{Tel Aviv University}
\author{Amitai Uzrad$^*$}
\affil{Tel Aviv University}

\date{\empty}

\begin{titlepage}
\def\thepage{}
\maketitle

\begin{abstract}
In their breakthrough ICALP'15 paper, Bernstein and Stein presented an algorithm for maintaining a $(3/2+\eps)$-approximate maximum matching in fully dynamic {\em bipartite} graphs with a {\em worst-case} update time of $O_\eps(m^{1/4})$; we use the $O_\eps$ notation to suppress the $\eps$-dependence.
Their main technical contribution was in presenting a new type of bounded-degree subgraph, which they named an {\em edge degree constrained subgraph (EDCS)}, which contains a large matching --- of size that is smaller than the maximum matching size of the entire graph by at most a factor of $3/2+\eps$.
They demonstrate that the EDCS can be maintained with a worst-case update time of $O_\eps(m^{1/4})$, and their main result follows as a direct corollary.
In their followup SODA'16 paper, Bernstein and Stein generalized their result for general graphs, achieving the same update time of $O_\eps(m^{1/4})$, albeit with an amortized rather than worst-case bound.
To date, the best {\em deterministic} worst-case update time bound for {\em any} better-than-2 approximate matching is $O(\sqrt{m})$ [Neiman and Solomon, STOC'13], [Gupta and Peng, FOCS'13]; allowing randomization (against an oblivious adversary) one can achieve a much better (still polynomial) update time for approximation slightly below 2 [Behnezhad, {\L}acki and Mirrokni, SODA'20]. 
%{\L}{\k{a}}cki 

In this work we\footnote{\em quasi nanos, gigantium humeris insidentes} simplify the approach of Bernstein and Stein for bipartite graphs, which allows us to generalize it for general graphs while maintaining the same bound of $O_\eps(m^{1/4})$ on the {\em worst-case} update time. Moreover, our approach is {\em density-sensitive}: If the {\em arboricity} of the dynamic graph is bounded by $\alpha$ at all times, then the worst-case update time of the algorithm is $O_\eps(\sqrt{\alpha})$.

\medskip
\noindent
	\textbf{Recent related work:} Independently and concurrently to our work, Roghani, Saberi and Wajc [arXiv'21] obtained two dynamic algorithms for
	approximate maximum matching with worst-case update time bounds. 
	Their first algorithm achieves approximation factor slightly better than 2 within $O(\sqrt{n} \cdot m^{1/8})$ update time,
	and their second algorithm achieves approximation factor $(2+\eps)$ for any $\eps > 0$ within $O_\eps(\sqrt{n})$ update time. 
	In terms of techniques, the two works are  entirely disjoint.

\end{abstract}
\end{titlepage}

\pagenumbering {arabic}
\input{intro}

\input{techOver}
\input{intro_final}
\input{prelim}

\input{PartI}
\input{PartII}
\input{PartIII_v2}
%\clearpage
%\bibliographystyle{latex8}
%\bibliography{randomMMbibfile}

%\clearpage
%$\\$
%$\\$
%\pagenumbering{roman}
%\appendix
%\centerline{\LARGE\bf Appendix}

\newpage
\bibliographystyle{alpha}
\bibliography{dynmatch,general}
\end{document}

%% file: latexmac.tex
\def\denseformat{
\setlength{\textheight}{9.5in}
\setlength{\textwidth}{6.9in}
\setlength{\evensidemargin}{-0.3in}
\setlength{\oddsidemargin}{-0.3in}
\setlength{\headsep}{10pt}
\setlength{\topmargin}{-0.44in}
\setlength{\columnsep}{0.375in}
\setlength{\itemsep}{0pt}
}
\newtheorem{theorem}{Theorem}[section]

\newtheorem{lemma}[theorem]{Lemma}

\newtheorem{corollary}[theorem]{Corollary}

\def\boldhead#1:{\par\vskip 7pt\noindent{\bf #1:}\hskip 10pt}
\def\ithead#1:{\par\vskip 7pt\noindent{\it #1:}\hskip 10pt}

\def\inline#1:{\par\vskip 7pt\noindent{\bf #1:}\hskip 10pt}
\def\midinline#1:{\par\noindent{\bf #1:}\hskip 10pt}
\def\dnsinline#1:{\par\vskip -7pt\noindent{\bf #1:}\hskip 10pt}
\def\ddnsinline#1:{\newline{\bf #1:}\hskip 10pt}
\def\largeinline#1:{\par\vskip 7pt\noindent{\large\bf #1:}\hskip 10pt}
\long\def\commhide #1\commhideend{}
\long\def\commtim #1\commendt{#1}
\long\def\commb #1\commbend{}
\long\def\commedit #1\commeditend{} % Editing comments, marked also by $>>>$

\long\def\commB #1\commBend{}       % Omit in 1996 (both TR and Siena)
\long\def\commex #1\commexend{}     % LN home exercise (hide solutions)

\long\def\commsiena #1\commsienaend{}  % omit in Siena, show in TR

\long\def\commBI #1\commBIend{}  % omit in Bar-Ilan

\long\def\CProof #1\CQED{}
\def\qed{\mbox{}\hfill $\Box$\\}
\def\blackslug{\hbox{\hskip 1pt \vrule width 4pt height 8pt
    depth 1.5pt \hskip 1pt}}
\def\QED{\quad\blackslug\lower 8.5pt\null\par}
\long\def\PPP#1{\noindent{\bf Proof:}{ #1}{\quad\blackslug\lower 8.5pt\null}}

\long\def\denspar #1\densend
{#1}
\newif\ifnotesw\noteswtrue% T to show box & marginal notes; F supresses.
   {\ifnotesw\marginpar[\hfill\(\top\)]{\(\top\)}\fi}%
   {\ifnotesw\marginpar[\hfill\(\bot\)]{\(\bot\)}\fi}

\newcommand{\mnote}[1]%
    {\ifnotesw\marginpar%
        [{\scriptsize\it\begin{minipage}[t]{\marginparwidth}
        \raggedleft#1%
                        \end{minipage}}]%
        {\scriptsize\it\begin{minipage}[t]{\marginparwidth}
        \raggedright#1%
                        \end{minipage}}%
    \fi}

\def\MathF{\hbox{\rm I\kern-2pt F}}
\def\MathP{\hbox{\rm I\kern-2pt P}}
\def\MathR{\hbox{\rm I\kern-2pt R}}
\def\MathZ{\hbox{\sf Z\kern-4pt Z}}
\def\MathN{\hbox{\rm I\kern-2pt I\kern-3.1pt N}}
\def\MathC{\hbox{\rm \kern0.7pt\raise0.8pt\hbox{\footnotesize I}
\kern-4.2pt C}}
\def\MathQ{\hbox{\rm I\kern-6pt Q}}

\newsavebox{\ttop}\newsavebox{\bbot}

\def\eps{\epsilon}

\def\polylog{\mbox{polylog}}

\def\nin{{~\not \in~}}

%% file: intro.tex
\section{Introduction}
Dynamic matching algorithms have been subject to   extensive research attention for more than a decade, starting with the pioneering work
of Onak and Rubinfeld \cite{OnakR10}. We do not aim to cover here the entire literature, but rather to briefly survey most of the state-of-the-art results that concern approximation factor 2 or less --- indeed the 2-approximation barrier is a central one in the context of graph matching. 
One may try to optimize the {\em amortized} (i.e., average) update time of an algorithm or its {\em worst-case}
(i.e., maximum) update time, over a worst-case sequence of graphs,
and we will put special emphasis on this distinction.  %amortized versus worst-case bounds.
Indeed, there is a strong separation between the state-of-the-art amortized versus worst-case time bounds for dynamic matching algorithms; a similar separation exists for various other dynamic graph problems, such as spanning tree, minimum spanning tree and two-edge connectivity.
%Next, we provide a brief literature survey on dynamic matchings. (See [24, 4, 26, 27] for a detailed survey.)
%Moreover, we shall also distinguish between deterministic versus randomized algorithms, and further classify randomized algorithms into two %types, depending on whether we make the {\em oblivious adversary assumption} or not; more on that below. 

Before starting the background survey, we summarize the key question around which our work revolves:
%\shay{could you please put questions and theorems inside grey-shaded boxes?}
\begin{Question} \label{q1}
Is there a {\em deterministic} algorithm for maintaining a better-than-2 (approximate) {\em maximum cardinality matching (MCM)} with a low worst-case update time?
Further, is it possible to push the approximation factor well below 2 (possibly using a randomized algorithm)?
\end{Question}

%\chris{Is there a reason why we are motivating randomized algorithms, if everything we are doing is deterministic?}
%\shay{the only prevoius better-than-2 approx is via a randomized algs; so it makes sense to separate into two goals (deterministic with %better-than-2 approx, or possibly randomized but with much-better-than-2 approx, and we happen to achieve both goals with a deterministic %alg, which only makes our results stronger}

\paragraph{Maximal matching.~}
In their seminal paper, Baswana et al.\ \cite{BGS11} showed that a maximal matching, which provides a 2-MCM, can be maintained with an {\em amortized}  update time of $O(\log n)$.
A constant amortized update time was given by Solomon \cite{Sol16}, and a {\em worst-case} $\polylog(n)$ update time was given by Bernstein et al.\ \cite{BFH19}. %; in both these results the update time holds in expectation and with high probability.
All these algorithms are randomized, and as common in the area of dynamic graph algorithms, they operate under the {\em oblivious adversarial model}, in which the adversary (the entity adding/deleting edges to/from the graph) may know all the edges in the graph
and their arrival order, as well as the algorithm to be used, but is not aware of the random bits used by
the algorithm, and so cannot choose updates adaptively in response to the randomly guided choices of the algorithm.
%this is a standard model that has been extensively used for analyzing randomized data-structures such
%as universal hashing [13] and dynamic connectivity [19].

Neiman and Solomon \cite{NS13} gave a deterministic algorithm for maintaining a maximal matching with a worst-case
update time of $O(\sqrt{m})$, where $m$ is the (dynamically changing) number of edges in the graph.
This $\sqrt{m}$ bound remains the state-of-the-art for dynamic maximal matching,
even allowing amortization and even allowing randomization under a non-oblivious adversary.

\paragraph{Crossing the 2-approximation barrier.~}
%In their seminal paper, Baswana et al.\ \cite{BGS11} gave a randomized algorithm for maintaining a maximal matching, and thus a 2-approximate {\em maximum cardinality matching (MCM)}, with an {\em amortized} update time of $O(\log n)$.
%A constant amortized update time was achieved by Solomon \cite{Sol16}, and a {\em worst-case} polylog update time was established by
%Bernstein et al.\ \cite{BHF19}. 
Neiman and Solomon \cite{NS13} were the first to cross the 2-approximation barrier: $3/2$-MCM   with a worst-case deterministic update time of $O(\sqrt{m})$. Soon afterwards Gupta and Peng \cite{GP13} showed that a $(1+\eps)$-MCM can be maintained with a worst-case deterministic update time of $O(\sqrt{m}/\eps^2)$, for any $0 < \eps < 1/2$; moreoever, if the maximum degree is always upper bounded by $\Delta$,
then a worst-case update time of $O(\Delta / \eps^2)$ can be achieved \cite{GP13}.

Bernstein and Stein \cite{BernsteinS15} showed that a $(3/2+\eps)$-MCM can be maintained in {\em bipartite} graphs with a {\em worst-case} deterministic update time of $O(m^{1/4} \eps^{-2.5})$. 
Their main technical contribution was in presenting a new type of bounded-degree subgraph,  an {\em edge degree constrained subgraph (EDCS)}, which contains a large matching --- of size that is smaller than the maximum matching size of the entire graph by at most a factor of $3/2+\eps$.
They demonstrate that the EDCS can be maintained with a worst-case update time of $O(m^{1/4} \eps^{-2.5})$ (and also with a worst-case {\em recourse}, which is the number of edge changes per update step, of $O(1/\eps)$), and their main result follows by running the bounded degree version of Gupta-Peng algorithm \cite{GP13} on top of the EDCS.
In their followup paper \cite{BernsteinS16}, they generalized their result for general graphs, using an inherently different
algorithm for the EDCS maintenance; they achieved the same update time of $O_\eps(m^{1/4})$, albeit with an amortized rather than worst-case bound.

%These results by \cite{BernsteinS15} and {BernsteinS16} demonstrate a se

Behnezhad et al.\ \cite{BLM20} showed that a (slightly better-than-2)-MCM can be maintained with arbitrarily small polynomial worst-case update time, that is, the update time of their algorithm is $O(\Delta^\eps + \polylog(n))$ and their approximation  is $2- \Omega_\eps(1)$, where $\Omega_\eps(1)$ is a tiny constant that shrinks rapidly as $\eps$ reduces.
This result is achieved via a randomized algorithm that assumes an oblivious adversary; for bipartite graphs, similar (but somewhat weaker) results can be achieved
without making the oblivious adversary assumption: Via a randomized algorithm against an adaptive adversary \cite{Wajc20,BhattacharyaHN16} and via a deterministic algorithm but with an amortized update time bound \cite{BK21,BhattacharyaHN16}.

% is a rather strong separation between the state-of-the-art results for bipartite and non-bipartite graphs.
%We focus here mostly on the results for general graphs
\paragraph{Bounded arboricity graphs.~}
The {\em arboricity} of an $m$-edge graph is the minimum number of forests into which it can be decomposed, and it ranges from 1 to $\sqrt{m}$.
%Equivalently,  the graph arboricity is $\alpha$ if for every $S\subseteq V$, it holds that $\frac{m_s}{n_s-1}\leq \alpha$, where $m_s$ and $n_s$ are the number of edges and vertices in the subgraph induced by $S$, respectively.
%\end{definition}
The family of bounded arboricity graphs can be viewed as the family of ``sparse everywhere'' graphs, containing bounded-degree graphs, all minor-closed graph classes (e.g., planar graphs and graphs of bounded treewidth), and randomly generated preferential attachment graphs. Moreover, many natural and real world graphs, such as the world wide web graph, social networks and transaction networks, are believed to have bounded arboricity. Thus, it is only natural to try and come up with better dynamic matching algorithms for bounded arboricity graphs.

Bernstein and Stein \cite{BernsteinS15} showed that by using a {\em weighted} EDCS, 
one can a maintain a $(1+\eps)$-MCM in {\em bipartite} graphs of arboricity at most $\alpha$ with a {\em worst-case} update time of 
$O(\alpha(\alpha + \log n) + \eps^{-4}(\alpha + \log n) + \eps^{-6})$.
In their followup paper \cite{BernsteinS16}, Bernstein and Stein generalized their result for general graphs using an ordinary EDCS, achieving a similar update time of $O(\alpha(\alpha + \log n + \eps^{-2}) + \eps^{-6})$, albeit with an amortized rather than worst-case bound and with approximation factor $3/2+\eps$ rather than $1+\eps$.
Peleg and Solomon \cite{PS16} showed that a $(1+\eps)$-MCM can be maintained in graphs of arboricity at most $\alpha$ with a {\em worst-case} update time of $O(\alpha / \eps^2)$; since $\alpha$ ranges between 1 and $\sqrt{m}$, this result generalizes the result of \cite{GP13} for general graphs with update time $O(\sqrt{m}/\eps^2)$.

\subsection{Our contribution}
In this work we simplify the approach of Bernstein and Stein \cite{BernsteinS15} for bipartite graphs, which allows us to generalize it for general graphs while maintaining the same upper bound  on the {\em worst-case} update time. 

%\chris{Shaded environment only for theorems, as of now.}
\begin{theo} \label{thm1}
For any dynamic graph $G$ subject to edge updates and for any $\eps < 1/2$, one can maintain a $(3/2+\eps)$-MCM for $G$
with a deterministic worst-case update time of  $O(m^{1/4} \eps^{-2.5} + \eps^{-6})$.
\end{theo}
Note that \Cref{thm1} resolves \Cref{q1} in the affirmative.
In particular, it provides the first {\em deterministic} algorithm achieving approximation better-than-2 with a worst-case update time that is strictly sublinear in the number of vertices for general graphs; it also provides the first (possibly randomized) algorithm achieving approximation well below 2 with such a worst-case update time bound.

Importantly, our approach is {\em density-sensitive}: If the {\em arboricity} of the dynamic graph is bounded by $\alpha$ at all times, then the worst-case update time of the algorithm reduces    to $O(\sqrt{\alpha} \eps^{-2.5} + \eps^{-6})$.
\begin{theo} \label{thm2}
Fix any parameter $\alpha \ge 1$.
For any dynamic graph $G$ subject to edge updates whose arboricity is upper bounded by $\alpha$ at all times and for any $\eps < 1/2$, one can maintain a $(3/2+\eps)$-MCM for $G$
with a deterministic worst-case update time of  $O(\sqrt{\alpha} \eps^{-2.5} + \eps^{-6})$.
\end{theo}
We note that the previous state-of-the-art arboricity-dependent update time, due to Peleg and Solomon \cite{PS16}, is $O(\alpha/\eps)$ --- which is quadratically higher than the update time provided by \Cref{thm2}, ignoring the $\eps$-dependence.
Although the approximation guarantee provided by the matching of \cite{PS16} is $1+\eps$ rather than $3/2+\eps$, no better arboricity-dependent update time bounds were known prior to this work, even for much larger approximation guarantee and even for amortized bounds.
%The update time provided by the algorithm of \Cref{thm2} is quadratically better than the previous state-of-the-a

%% file: techOver.tex
\subsection{Technical Overview}
%\shay{what's missing perhaps (not sure): explain why BS15 argument is for bipartite graphs...}
Our argument for maintaining an EDCS efficiently consists of the following three steps.

Suppose that the maximum degree in the dynamic graph never exceeds $\Delta$.
In the first step, we demonstrate that an EDCS can be maintained with a worst-case update time of $O_\eps(\Delta)$;
we use the $O_\eps$ notation to suppress the $\eps$-dependence.
Our key observation for this step is that the argument used by \cite{BernsteinS15} for the case of {\em bipartite graphs of bounded arboricity} can be employed --- with minimal changes --- for {\em general graphs}, to restore a valid EDCS following any edge update by 
computing and augmenting a short (possibly non-simple) alternating path in the EDCS.
% following any edge update to restore a valid EDCS.
The changes that we introduce only simplify the argument, as we only trim parts from it: (1) While the argument of \cite{BernsteinS15} 
proves that the alternating path is simple, which is true only in bipartite graphs, we
 % don't claim that the alternating path is simple and 
make do without relying on simple paths, which is crucial for coping with general graphs. 
(2) We don't maintain a dynamic edge orientation.  

In the second step, we reduce the update time from $O_\eps(\Delta)$ to $O_\eps(\sqrt{\Delta})$.
Here too we use an idea from \cite{BernsteinS15}, of scanning just a small subset of neighbors following a change of degree (in the EDCS), which then triggers a discrepency between the true degrees and their estimations by their neighbors. 
The main difference to \cite{BernsteinS15} is that we do not try to preserve an EDCS with the same parameters as those achieved in step 1, but rather allow them to degrade by a small (constant) factor. This has no effect whatsoever on any of the guarantees, yet it simplifies the algorithm and its analysis quite a bit.

Finally,  we demonstrate that $\Delta$ can be substituted with either $\Theta(\sqrt{m} /\eps)$ or $\Theta(\alpha / \eps)$, where $\alpha$ is a fixed upper bound on the arboricity of the dynamic graph. This follows easily from the following result:
\begin{theorem} \cite{Sol18} \label{basicsparse}
Let $G$ be a graph with arboricity $\alpha$.
Suppose that each vertex $v$ in $G$ ``marks'' (up to) $\delta := c(\alpha / \eps)$ arbitrary incident edges, for an appropriate constant $c$.
%if a vertex has less than $\delta$ incident edges, it marks all of them.
The graph $G'$ obtained as the union of all edges marked twice (by both endpoints) is a {\em $(1+\eps)$-MCM sparsifier} for $G$,
i.e., $\mu(G) \le (1+\eps)\mu(G')$, where $\mu(H)$ denotes the maximum matching size of any graph $H$.
\end{theorem} 
It is straightforward to dynamically maintain the graph $G'$ defined by \Cref{basicsparse} within constant worst-case update time.
The only subtlety arises in case we would like to substitute $\Delta$ with $\Theta(\sqrt{m} /\eps)$, {\em where $m$ dynamically changes over time}; this issue can be resolved quite easily, as we show in \Cref{partIII}.
This sparsification step is very simple, whereas if instead we   resort to dynamic edge orientations, as done by \cite{BernsteinS15}, life would  become   more complicated; the following was written in Section 5 of \cite{BernsteinS15}, and is mostly attributed to the interplay between dynamic edge orientations and the discrepency between the true and estimated degrees in the EDCS: ``The details, however, are quite
involved, especially since we need a worst-case update time.''

From the above discussion, it is clear that each of the steps above simplify a respective ingredient from the argument of \cite{BernsteinS15}.
Another source of simplification stems from the fact that our argument is broken into three rather separate steps, which are then combined together in a rather natural way, 
which stands in contrast to the argument of \cite{BernsteinS15}, in which all the ingredients are intertwined together.
%\shay{the next para perhaps doesn't belong here}
Interestingly, our improved bounds over \cite{BernsteinS15,BernsteinS16} are achieved as a direct by-product of this simplification.
%The reason we emphasized the simplicity of our argument so much is that this is what allowed us to achieve the improvements. 
%to geneeralize the argument for general graphs while maintaining the same bound of $O_\eps(m^{1/4})$ on the worst-case update time as that %of \cite{BernsteinS15} for bipartite graphs.
%Moreover, our density-sensitive update time bound, namely $O_\eps(\sqrt{\alpha})$, is significantly better than the respective upper bound of 
%$O(\alpha(\alpha + \log n))$ by \cite{BernsteinS15}.

%\shay{maybe this is pushing it too much? I don't want to irritate the authors of \cite{BernsteinS15}}
%\chris{I think we need a "Our Results" section}

%% file: intro_final.tex
\paragraph{Related work.~}
The EDCS was introduced in \cite{BernsteinS15} for dynamic matching algorithms, but it turned out to be a very useful graph structure also outside the area of dynamic graph algorithms. In particular, EDCSs are especially useful in space and communication constrained settings, see e.g.~\cite{AsB19,ABBMS19,Ber20}. %\shay{add refs}

%To improve on the update times from EDCS-based algorithms, one typically has to assume incremental or decremental settings. Here 
For any fixed $\eps > 0$, a $(1+\eps)$-MCM can be maintained in constant (respectively, $\polylog(n)$) amortized update time
in incremental (resp., decremental) graphs \cite{FLSSS19} (resp., \cite{BK21}).
Importantly, these results only concern amortized bounds, and no better worst-case time bounds for incremental or decremental graphs
are known than the aforementioned results.

The improvement to the update time in bounded arboricity graphs versus general graphs becomes less  significant as the arboricity grows.
Milenkovi\'{c} and Solomon \cite{MS20} showed that a $(1+\eps)$-MCM can be maintained with a worst-case update time of 
$O(\frac{\beta}{\eps^3}\log\frac{1}{\eps})$, where $\beta = \beta(G)$ 
is the {\em neighborhood independence number} of the graph $G$, i.e., the size of the largest independent set in the neighborhood of any vertex.
Graphs with bounded neighborhood independence, already for constant $\beta$, constitute a wide family of {\em possibly dense graphs}, including line graphs, unit-disk graphs and graphs of bounded growth. 

%\shay{mention work on dynamic matching in other graph classes;
%in general, don't spend too much time, should keep it short}

\paragraph{Concurrent work.~} Independently and concurrently to our work, Roghani, Saberi and Wajc \cite{RSW21} obtained two dynamic algorithms for
	approximate maximum matching with worst-case update time bounds. 
	Their first algorithm achieves approximation factor slightly better than 2 within $O(\sqrt{n} \cdot m^{1/8})$ update time,
	and their second algorithm achieves approximation factor $(2+\eps)$ for any $\eps > 0$ within $O_\eps(\sqrt{n})$ update time. 
	In terms of techniques, the two works are  entirely disjoint.  
	
\paragraph{Organization.~}  
%The remainder of the paper consists of three main parts.
In \Cref{prel} we present the basic notation and preliminaries used throughout. \Cref{partI} describes how to achieve an $O_{\eps}(\Delta)$ deterministic worst-case update time, where $\Delta$ is an upper bound on the maximum degree in the dynamic graph. In \Cref{partII} we improve this update time to $O_{\eps}(\sqrt{\Delta})$.
Finally, in~\Cref{partIII}, we combine the algorithm from~\Cref{partII} with a simple sparsification technique to reduce the update time to $O_\eps(m^{1/4})$, or, in case of graphs with arboricity bounded by $\alpha$, to $O_\eps(\sqrt{\alpha)}$.

%% file: prelim.tex
\section{Preliminaries} \label{prel}

\noindent Let $G = (V,E)$ be an undirected, unweighted graph, where $|V|=n$ and $|E|=m$. Let $H = (V,E_H)$ be a subgraph of $G$, with $E_H \subseteq E$; denote by $d_H(v)$ the {\em degree} of vertex $v$ in the subgraph $H$.
For technical convenience, we shall use a weight function over the edges of $G$, $w: E \rightarrow \mathbb{N}$, 
which assigns a nonnegative integer for each edge of the graph; the weights of edges in $H$ are inherited from their weights in $G$.  
Specifically, the weight of any edge $(u,v) \in E$, denoted by $w(u,v)$, is given by $d_H(u)+d_H(v)$.
Following \cite{BernsteinS15,BernsteinS16}, we say that $H$ is an {\em $(\beta,\beta^-)$-EDCS} for $G$, for any pair $\beta, \beta^-$ of real numbers such that $\beta^- < \beta$, if $H$ is a spanning subgraph of $G$ that satisfies the following two properties: 
%subset $H$ that the following two properties hold:
 
\begin{enumerate}
	\item[(\textbf{P1})] For each $(u,v) \in H$, $w(u,v) \leq \beta$.
	\item[(\textbf{P2})] For each $(u,v) \notin H$, $w(u,v) \geq \beta^-$.
\end{enumerate}

For a vertex $v$, denote by $N(v)$ its adjacency list.
Denote by $\mu(G)$ the maximum matching size of a graph $G$. We say that a subgraph $G'$ of $G$ is a {\em $(1+\varepsilon)$-approximate matching sparsifier} of $G$ if $\mu(G) \le (1+\eps)\mu(G')$.
For any sequence of edge updates that defines a dynamic graph $G$ and for any sparsifier $G'$ for $G$, we say that an algorithm maintaining $G'$ has a worst-case {\em recourse} (or a worst-case \emph{update ratio}) of $r$ if, {\em for each} edge update in $G$, the number of edge changes made to $G'$ is at most $r$. 
%alternatively, we may say in this case that the worst-case {\em update ratio} of $G'$ with respect to $G$ is $r$. 

Finally, we rely on the following key property that relates the EDCS to $\mu(G)$.

\begin{lemma}[Lemma 2.5 of~\cite{AsB19}, proved first by~\cite{BernsteinS16,BernsteinS15}] \label{keyl}
Let $G = (V, E)$ be any graph and $\eps < 1/2$ be a parameter. For any parameters $\lambda \leq \frac{\eps}{100}$, $\beta\geq 32\lambda^{-3}$, and $\beta^{-} \geq (1-\lambda)\cdot \beta$, in any subgraph $H := EDCS(G, \beta, \beta^{-})$, it holds that $\mu(G)\leq \left(\frac{3}{2}+\eps\right)\mu(H)$.
\end{lemma}

%\shay{need to define matching sparsifier. Need to define $\mu(G)$ as in BS15}

%\shay{may need to define the recourse bound; the update ratio}

%\shay{need to add here Lemma 2.5 from the SOSA19 paper of Assadi and Bernstein, and give credit to BS15 and BS16 for it}
%\noindent Where $\beta$ and $\epsilon$ are given, and we assume that $\epsilon\beta > 1$.

%\chris{done, remove this comment (unless we need to change more things)}

%% file: PartI.tex
\section{Maintaining an EDCS, Part I: Update time $O_\eps(\Delta)$} \label{partI}
In this section we show that an $(\beta,(1-\eps)\beta)$-EDCS can be maintained dynamically with a worst-case update time of $O(1/\eps + \Delta)$.
This section follows along similar lines as those in Appendix C.1 of \cite{BernsteinS15}.

\noindent To dynamically maintain an EDCS, it is instructive to define the following two types of edges: 

\begin{enumerate}
	\item[\textbf{(P1)}] 	\hypertarget{P1}{}
	A \emph{full} edge $(u,v)$ is in $H$ and satisfies $w(u,v) = \beta$.
	\item[\textbf{(P2)}] 	\hypertarget{P2}{}
	A \emph{deficient} edge $(u,v)$ is not in $H$ and satisfies $w(u,v) = (1-\epsilon)\beta$.
\end{enumerate}

%\noindent 
Upon insertion of edge $(u,v)$ in $G$, if $w(u,v) \geq (1-\epsilon)\beta$, then we do not add the edge to $H$, and Properties 
(\hyperlink{P1}{P1}) and (\hyperlink{P2}{P2}) continue to hold. On the other hand, if $w(u,v) < (1-\epsilon)\beta$, we need to add $(u,v)$ to $H$. Doing so will increase $d_H(u)$ and $d_H(v)$ by 1, which may lead to a violation of Property (\hyperlink{P1}{P1}) for other edges in $H$ that are incident to either $u$ or $v$. We will then fix an arbitrary such violating edge (if any) in a similar way, and thus we proceed to finding and augmenting an \emph{alternating path} in $H$ (edges along the path are in $H$ and out of $H$, alternately), which ultimately allows us to add $(u,v)$ to $H$ and still maintain Properties (\hyperlink{P1}{P1}) and (\hyperlink{P2}{P2}) for all vertices. 
Next, we describe this process in detail.

%\bigskip

%\noindent 
We say that a vertex $x$ is \emph{increase-safe} (respectively, \emph{decrease-safe}) if it has no incident full (resp., deficient) edges.
%, and say that it is \emph{decrease-safe} if it has no incident deficient edges. 
Returning to adding $(u,v)$ to $H$, let us focus on vertex $v$; we will then treat vertex $u$ analogously. If $v$ is increase-safe, then, by definition, adding $(u,v)$ to $H$ does not violate Property (\hyperlink{P1}{P1}) for any edge incident on $v$.
%, since there are no full edges incident to $v$. 
If $v$ is not increase-safe, then it must have at least one incident full edge incident on it, say $(v,p_1)$. We would like to add $(u,v)$ to $H$ and remove $(v,p_1)$ from $H$, thereby leaving $v$'s degree (in $H$) unchanged. Doing so would decrease $d_H(p_1)$, which we can do only if $p_1$ is decrease-safe. If $p_1$ is decrease safe, then adding $(u,v)$ to $H$ and removing $(v,p_1)$ leaves $v$'s degree unchanged, decreases $p_1$'s degree and reestablishes Properties (\hyperlink{P1}{P1}) and (\hyperlink{P2}{P2}) for all vertices (except possibly $u$). However, if $p_1$ is not decrease-safe, it must have an incident deficient edge, say $(p_1,p_2)$. We can add this edge to $H$ and continue from $p_2$, just like we did from $v$. We can continue in this manner, stopping when we find either an increase-safe or a  decrease-safe vertex. We will prove below that this process terminates, and when it does, all the degrees (in $H$) have returned to their value prior to the edge insertion, except for at most two vertices, which we will need to handle separately. We remark that the alternating path that we augment throughout the process is not necessarily a simple path (a simple path is obtained only in bipartite graphs),
but this fact has no consequences whatsoever on the validity of the argument. An alternating path from $v$ in the case of edge insertion can be seen in Figure 1:
%\shay{I edited and also added a few sentences in the end that we'll need to polish}

\begin{figure}[!htb]
  \center{\includegraphics[scale=0.65]{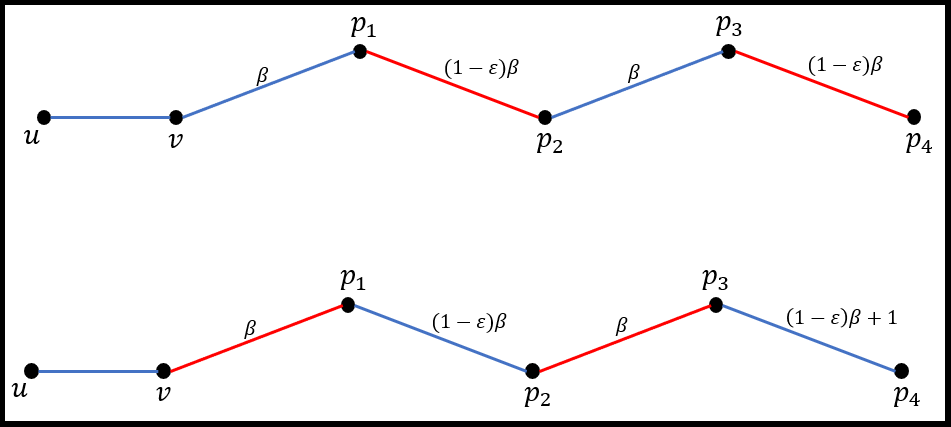}}
  \caption{\label{fig:1}$v$'s alternating path upon insertion of edge $(u,v)$ to $H$. Blue edges designate edges in the $EDCS$, and the red ones are not. On the top we have the path at the beginning of the alternating process, and on the bottom we have the path at the end. Notice that $w(p_3,p_4)$ grows by one since $d_H(p_4)$ grows by one.}
\end{figure}

%\bigskip

%\noindent 
Upon deletion of edge $(u,v)$ from $G$, we can reestablish Properties (\hyperlink{P1}{P1}) and (\hyperlink{P2}{P2})  via a symmetric process. If $(u,v)$ is not in $H$, then we do not need to change $H$. If $(u,v)$ is in $H$ and both $u$ and $v$ are decrease-safe, we just remove the edge $(u,v)$. Otherwise, we find an alternating path in the same way that we did for an edge insertion. For completeness, we provide the pseudo-codes for  edge insertion and for edge deletion  at the end of this section (where $E_H$ is the set of edges that belong to $H$). Along with the code for insert and delete, two auxiliary codes are added, which are used to find the next full/deficient edge on the path and update the relevant data structures.

%\bigskip

Observe that any alternating path that is found and augmented following an edge update is a path that alternates between full and deficient edges. The following lemma and its proof are similar (but not identical) to Lemmas 7 and 10 from \cite{BernsteinS15} and their respective proofs. 
\begin{lemma} \label{mainl}
For any path $P$ of alternating full and deficient edges, we have $|P| \le \frac{2}{\epsilon}$.
\end{lemma} 
\begin{proof}
Consider first the case that $P=(p_0,p_1,...,p_k)$ starts with a full edge. Let $d=d_H(p_0)$.
Since $(p_0,p_1)$ is full, $d_H(p_1) = \beta-d$. 
%Moreover,  but since $p_0$ is incident on a full edge, we have $d \le \beta-1$.
Since $(p_1,p_2)$ is deficient, $d_H(p_2)=(1-\epsilon)\beta-d_H(p_1)=d-\epsilon\beta$. Continuing, we get that $d_H(p_3)=\beta-d+\epsilon\beta$, $d_H(p_4)=d-2\epsilon\beta$, and in general we inductively get $d_H(p_{2i})=d-i\epsilon\beta$, for any $0 \le i \le k/2$. 
Since each vertex in $P$ has an incident full edge, all vertices have a positive degree, and thus $d_H(p_{2i}) = d-i\epsilon\beta >0$.
Since $d \leq \beta$, it follows that $|P| \le \frac{2}{\epsilon} -1$. If $P$ starts with a deficient edge, let $d=d_H(p_1)$.
Since $(p_1,p_2)$ is full, $d_H(p_2) = \beta-d$. 
%Moreover,  but since $p_0$ is incident on a full edge, we have $d \le \beta-1$.
Since $(p_2,p_3)$ is deficient, $d_H(p_3)=(1-\epsilon)\beta-d_H(p_2)=d-\epsilon\beta$. Continuing, we get that $d_H(p_4)=\beta-d+\epsilon\beta$, $d_H(p_5)=d-2\epsilon\beta$, and in general we inductively get $d_H(p_{2i+1})=d-i\epsilon\beta$, for any $0 \le i < k/2$. 
Since all vertices have a positive degree, we have $d_H(p_{2i+1}) = d-i\epsilon\beta >0$.
Since $d \leq \beta$, it follows that $|P| \le \frac{2}{\epsilon}$.
\qed
\end{proof}

After finding and augmenting an alternating path starting at $v$ as described above, we repeat the same update procedure starting at $u$ (we execute this update procedure sequentially).

\noindent {\bf Remark.}
%\begin{enumerate}[topsep=2pt,itemsep=1pt,partopsep=1ex,parsep=1ex]
%	\item 
When $(u,v)$ is inserted to $H$ --- although no specific changes are needed --- it is instructive to consider the case where the alternating path starting at $v$ intersects $u$. (Recall that we first handle $v$ and only later $u$.) If $v$'s path reaches $u$ with a full edge (which was the only possibility in the bipartite case), then we would be done since the degree in $H$ of each vertex would be the same as it was prior to the edge insertion, which 
reestablishes Properties (\hyperlink{P1}{P1}) and (\hyperlink{P2}{P2}) of an EDCS. (When we later handle $u$, nothing else will be done.)
If $v$'s path reaches $u$ with a deficient edge, then we just continue the path as we normally would --- the fact that $u$'s degree in $H$ is "momentarily" raised by 2 does not affect the correctness of the algorithm and its analysis; in such a case we still need to handle $u$ in the end, and the fact that the alternating path reached it does not change that. 
Symmetrically, when $(u,v)$ is deleted from $H$, the same claim holds as for insertion, only that here if $v$'s path reaches $u$ with a deficient edge then we would be done, and if $v$'s path reaches $u$ with a full edge, then we just continue the path  normally. 
%\end{enumerate}
%\vspace{2pt}

%\shay{you did not address the case where the alternating path starting at $v$ intersects $u$, and in particular when it terminates at $u$.
%Even if no specific changes are needed to handle this case, we still want to draw the readers' attention to this case}

We have thus shown the following:

\begin{lemma}
Following each edge update, we can reestablish Properties (\hyperlink{P1}{P1}) and (\hyperlink{P2}{P2})  using at most $\frac{4}{\epsilon}$ insertions and deletions to and from $H$.
\end{lemma}

Thus far we have explained how to find the alternating paths starting at $v$ and $u$ in lay terms, and also proved that they contain together at most $\frac{4}{\epsilon}$ edges. However, we still haven't described the data structures required for finding and augmenting these paths efficiently.  % and thus we still haven't addressed the running time of the update procedure.
In order to find the alternating paths, we need to maintain the necessary data structures to identify full and deficient edges. Each vertex $x$ will maintain its degree in $H$, $d_H(x)$, as well as a partition of its adjacent edges into three lists: (1) $F(x)$, list of $x$'s full adjacent edges, (2) $D(x)$, list of $x$'s deficient adjacent edges, and (4) $R(x)$, list of $x$'s remaining edges. 
Consider the process of constructing the alternating path and assume we are currently at vertex $x$. If we are looking for a full (resp., deficient) edge, we check $F(x)$ (resp., $D(x)$). %, and if we are looking for a deficient edge, we check $D(v)$. 
Each test takes $O(1)$ in the worst case (whether the list is empty or not). If the list is non-empty, we remove from it one arbitrary edge (say the first one on the list), again in $O(1)$ time. After we find a full (resp., deficient) edge and continue on our path, we must then remove the previous edge from the list $F(x)$ (resp., $D(x)$) and add it to $R(x)$, yet again in $O(1)$ time. Since we spend $O(1)$ time for processing each vertex along the path, computing and augmenting the entire alternating path (which also concerns updating $F,D$ and $R$ for the vertices along the way) can be done in total $O(\frac{1}{\epsilon})$ time. Moreover, at the end of each such update, we have (up to) 2 vertices whose degree in $H$ has changed, i.e., the final vertex of $v$'s alternating path and the final vertex in $u$'s alternating path. 
Consequently, for each neighbor $w$ of any of these (up to) 2 vertices whose degree in $H$ has changed,
%Denote the final vertex in $v$'s path by $p_k$; since the degree of $p_k$ in $H$ has changed, 
we spend $O(1)$ time to update the lists $F(w)$, $D(w)$ and $R(w)$ accordingly. Since the degree in $G$ is upper bounded by $\Delta$, the worst-case update time of this process is $O(\Delta)$. Summarizing, the total worst-case update time is $O(\Delta + \frac{1}{\epsilon})$. %which leads to $O_\epsilon(\Delta)$.

\begin{algorithm}
\caption{Insert $(u,v)$}
\begin{algorithmic}
\IF {$d_H(u)+d_H(v) < (1-\epsilon)\beta$}
\STATE $E_H \leftarrow E_H \cup \{(u,v)\}$
\STATE HandleFull $(v)$
\STATE HandleFull $(u)$
\ENDIF
\end{algorithmic}
\end{algorithm}

\begin{algorithm}
\caption{Delete $(u,v)$}
\begin{algorithmic}
\IF {$(u,v)\in E_H$}
\STATE $E_H \leftarrow E_H \setminus \{(u,v)\}$
\STATE HandleDeficient $(v)$
\STATE HandleDeficient $(u)$
\ENDIF
\end{algorithmic}
\end{algorithm}

\begin{algorithm}
\caption{HandleFull $(p_i)$}
\begin{algorithmic}
\IF {EMPTY($F(p_i)$) = FALSE}
\STATE $x \leftarrow \mbox{POP}(F(P_i))$ // i.e., we set $x$ as the first element of $F(p_i)$, and then remove $x$ from $F(p_i)$
\STATE $E_H \leftarrow E_H \setminus \{(p_i,x)\}$
\STATE HandleDeficient $(x)$
\ENDIF
\end{algorithmic}
\end{algorithm}

\begin{algorithm}
\caption{HandleDeficient $(p_i)$}
\begin{algorithmic}
\IF {EMPTY($D(p_i)$) = FALSE}
\STATE $x \leftarrow \mbox{POP}(D(p_i))$ // i.e., we set $x$ as the first element of $D(P_i)$, and then remove $x$ from $D(p_i)$
\STATE $E_H \leftarrow E_H \cup \{(p_i,x)\}$
\STATE HandleFull $(x)$
\ENDIF
\end{algorithmic}
\end{algorithm}

%two priority queues denoted $Q_v^{min}$ and $Q_v^{max}$ (implemented by using balanced binary search trees). In $Q_v^{min}$ we have the neighbors of $v$ in $G$ but that their connecting edge is \emph{not} in $H$, sorted by their degree in $H$ in such a way that we can extract the neighboring vertex of $v$ with \emph{minimum} degree in $H$ in logarithmic time. In $Q_v^{max}$ we have the neighbors of $v$ in $G$ that their connecting edge \emph{is} in $H$, sorted by their degree in $H$ in such a way that we can extract the neighboring vertex of $v$ with \emph{maximum} degree in $H$ in logarithmic time. 

%% file: PartII.tex
%\section{Improving the Worst Case Bound to $O_\epsilon(\sqrt{\Delta log(n)})$}
\section{Maintaining an EDCS, Part II: Improved Update Time $O_\eps(\sqrt{\Delta})$}
\label{partII}
%In this section we show that an EDCS can be maintained dynamically with a worst-case update time of $O(1/\eps + \Delta)$ \shay{I removed the %log n; btw, note that a slash is needed before log}.
In this section we achieve roughly a quadratic improvement over the worst-case update time achieved in~\Cref{partI}.
To this end, we tweak the algorithm of~\Cref{partI} slightly: instead of {\em notifying} all the neighbors of a vertex following a {\em permanent change} in its degree in $H$, we only notify a subset of its neighbors; if this vertex is at the end of an alternating path, we refer to its degree change as permanent, 
to distinguish from a momentary degree change of a vertex not in the end of that path --- though of course the degree of that vertex in $H$ might change again in subsequent update steps.
By ``notify'' a neighbor we mean that the algorithm updates all the relevant  data structures of that neighbor concerning the degree change.
Next, we explain this tweak in detail.
%of the last vertex in the alternating path of its degree change, and not all of them. 
For each vertex $x$, we will maintain a {\em cyclic queue} $Q_x$ that holds all its neighbors, and each time  
there is a permanent change in $x$'s degree in $H$ (since it is the last vertex of an alternating path), it will only notify the first $\frac{10\Delta}{\epsilon\beta}$ neighbors in $Q_x$ (so that they can update their relevant data structures), and then these neighbors will be moved to the end of $Q_x$;  new neighbors will be added to the end of $Q_x$ as well, and the point is that once a neighbor (old or new)  has been moved to the end of $Q_x$, all its data structures are up-to-date with respect to $x$'s current degree in $H$. Edge deletions incident on $x$ are also reflected in $Q_x$, by simply removing these vertices from $Q_x$, wherever they might be there (via appropriate pointers).
Clearly, this tweak reduces the worst-case update time from $O(\Delta + 1/\eps)$ to $O(\frac{\Delta}{\epsilon\beta} + 1/\eps)$, 
which we record in the next lemma for further use:
\begin{lemma} \label{tweakedl}
The worst-case update time of this tweaked algorithm is $O(\frac{\Delta}{\epsilon\beta} + 1/\eps)$.
\end{lemma}

It remains   to analyze the ramifications of notifying only a fraction of the neighbors of a vertex following its permanent degree change in $H$. 
For a vertex $x$, recall that $d_H(x)$ denotes $x$'s degree in $H$; for a neighbor $w$ of $x$, i.e., $w \in N(x)$, denote by $\tilde{d}^w_H(x)$ the
estimation that $w$ has on $x$'s degree in $H$.
%, and denote by $\bar{d}_H(x)$ the maximum estimation on $x$'s degree in $H$, taken over all neighbors of $x$. 
We are interested in upper bounding $\texttt{Dis}(H) := \max_{x \in V, w \in N(x)} \{d_H(x) - \tilde{d}^w_H(x)\}$, 
i.e., the maximum discrepancy, over all vertices, between the degree of a vertex in $H$ and its estimated degree by any of its neighbors. 
% over all vertices $x$ in the graph. 
%of a vertex in $H$ and the estimation of that degree by its neighbors, over all vertices.
%\shay{I started using the term permanent deg change without defining it; but we need to clarify that such a change may happen many times, so %it's not really permanent...}
The maximum degree in $G$ is $\le \Delta$, thus the number of vertices in the cyclic queue $Q_x$ of $x$ is at most $\Delta$. 
Since $Q_x$ is cyclic, and it updates one batch of $\frac{10\Delta}{\epsilon\beta}$ vertices at a time, it follows that $\frac{\epsilon\beta}{10}$ would be the maximum number of batches before $w$ gets notified by $v$ again by its degree change in $H$, hence
%, , the maximum difference between $x$'s actual degree in $H$ and its degree that is stored at its neighbor, say $w$, 
$\texttt{Dis}(H) \le  \frac{\epsilon\beta}{10}$.  %\shay{last sentence needs to be explained better...}. 

The following lemma shows that we can still upper bound the length of any alternating path by $O(1/\eps)$.
\begin{lemma} \label{recourse}
For any path $P$ of alternating full and deficient edges, we have $|P| \le \frac{5}{2\epsilon}$.
%Let P be a path of alternating full and deficient edges. Then the length of P is still $O(\frac{1}{\epsilon})$.
\end{lemma} 
%\shay{I didn't read this proof; but since I commented the respective proof of the previous section, I suppose the changes there will be %carried over here}
\begin{proof}
Consider first the case that $P=(p_0,p_1,...,p_k)$ starts with a full edge. Recall that ${\tilde{d}^{p_i}_H(p_{i+1})}$ is the estimation of $p_i$ on the degree of $p_{i+1}$ in $H$. Let $d=d_H(p_0)$. 
%and clearly $d \leq \beta$. 
Since $(p_0,p_1)$ is full "in the eyes" of $p_0$ (i.e., according to the degree estimation on $p_1$ that $p_0$ has), we get that $\tilde{d}^{p_0}_H(p_1) =\beta -d_H(p_0) =\beta-d$.
Since $\texttt{Dis}(H) \le  \frac{\epsilon\beta}{10}$, we get that
$d_H(p_1)$ and $\tilde{d}^{p_0}_H(p_1)$ differ by at most an additive term of $\frac{\epsilon\beta}{10}$,
i.e., $d_H(p_1) = \tilde{d}^{p_0}_H(p_1) \pm \frac{\epsilon\beta}{10}$;
in what follows we shall use the $\pm$ operator to ``absorb'' these  additive terms. 
Similarly, since $(p_1,p_2)$ is deficient in the eyes now of $p_1$,  
we get that $$\tilde{d}^{p_1}_H(p_2) ~=~ (1-\epsilon)\beta-d_H(p_1) ~=~ (1-\epsilon)\beta-\tilde{d}^{p_0}_H(p_1)\pm \frac{\epsilon\beta}{10} ~=~ (1-\epsilon)\beta-\beta+d\pm \frac{\epsilon\beta}{10}=d-\epsilon\beta \pm \frac{\epsilon\beta}{10}.$$ Continuing in the same way, we get   $\tilde{d}^{p_2}_H(p_3)=\beta-d+\epsilon\beta\pm \frac{2\epsilon\beta}{10}$, and then $\tilde{d}^{p_3}_H(p_4)=d-2\epsilon\beta\pm \frac{3\epsilon\beta}{10}$, and in general we inductively get $\tilde{d}^{p_{2i-1}}_H(p_{2i})=d-i\epsilon\beta \pm \frac{(2i-1)\epsilon\beta}{10}$, which implies that $d_H(p_{2i}) \le d-i\epsilon\beta+ \frac{i\epsilon\beta}{5}=d-\frac{4i\epsilon\beta}{5}$. Since each vertex in $P$ has an incident full edge, all vertices have a positive degree, and since $d \leq \beta$, it follows that $|P| \le \frac{5}{2\epsilon}-1$. If $P$ starts with a deficient edge, we can go through the same argument again (here too following similar lines as those in the proof of \Cref{mainl}),
to conclude that  $|P| \le \frac{5}{2\epsilon}$.
\qed
\end{proof}

\begin{corollary} \label{rec}
The worst-case recourse bound for maintaining $H$ is $O(1/\epsilon)$.
\end{corollary}

We next argue that, while Properties 
(\hyperlink{P1}{P1}) and (\hyperlink{P2}{P2}) may no longer hold ---
and thus $H$ may no longer provide a $(\beta, (1-\epsilon)\beta)$-EDCS ---
similar properties, for a slightly different choice of parameters, do hold.
%\shay{the following lemma requires revision, I can do that}
\begin{lemma} \label{crux}
Define $\gamma = \beta(1+\epsilon/10)$.
The  tweaked algorithm of this section maintains a $(\gamma,(1-2\epsilon)\gamma)$-EDCS, or in other words
it satisfies the following  two  properties:
\begin{enumerate}
	\item[(\textbf{P1'})] If $(u,v) \in H$, then $w(u,v) \leq \gamma$.
	\item[(\textbf{P2'})] If $(u,v) \notin H$, then $w(u,v) \geq (1-2\epsilon)\gamma$.
\end{enumerate}

\end{lemma} 

\begin{proof}
If $(u,v) \in H$, it no longer holds that $w(u,v) \le \beta$, due to the discrepancy $\texttt{Dis}(H) \le  \frac{\epsilon\beta}{10}$ between the degrees of vertices and the estimations of those by their neighbors.
Following a permanent increase in the degree of some vertex $u$ in $H$, it is increase-safe by definition, meaning that none of its adjacent edges in $H$ are full. However, the increase-safe condition is defined with respect to the estimated degree at the other endpoints, so in fact it only holds at this stage that 
$$w(u,v) ~=~ d_H(u) + \tilde{d}^u_H(v) ~\le~ d_H(u) + d_H(v) + \epsilon \beta / 10 \le \beta(1+\epsilon/10).$$
Symmetrically, it no longer holds that $w(u,v) \ge (1-\epsilon)\beta$, but rather $w(u,v) \ge (1-\epsilon)\beta - \epsilon \beta / 10 = \beta(1- 11\epsilon/10)$.
We next argue that $H$ provides a $(\gamma,(1-2\epsilon)\gamma)$-EDCS, for $\gamma = \beta(1+\epsilon/10)$.
Indeed, for any $(u,v) \in H$, we have $w(u,v) \le \beta(1+\epsilon/10) = \gamma$, whereas for any $(u,v) \nin H$, we have $w(u,v) \ge \beta(1- 11\epsilon/10) \ge \gamma(1-2\epsilon)$, where the last inequality holds for any $\epsilon>0$. 
\qed
\end{proof}

\input{putTogether.tex}

%% file: putTogether.tex
%\amitai{From here we need to take Lemma 2.5 from AB'19 and use a scaling argument to fit $\lambda$ accordingly, maybe should be done only %after PartIII? It's just a factor 2 for the $\epsilon$..}
\paragraph{Putting it all together.~}
We need to calibrate the parameters properly, so as to be able to apply \Cref{keyl}.
In \Cref{partI} we did not bother ourselves with such details, since the bound achieved there just served as a stepping stone towards the one achieved in this section; we simply mentioned there that we maintain a $(\beta,(1-\eps)\beta)$-EDCS.
Next, we will be more precise.
 Our aim here is to argue that the maintained EDCS $H$ is a $(3/2+\eps)$-approximate matching sparsifier for $G$, for a parameter $\eps > 0$; thus $\eps$ is now fixed.
We take $\lambda' = \lambda / 2 = \frac{\eps}{200}$ and we will take $\beta$ so that it is at least $32 \lambda^{-3}$.
Now the algorithm of \Cref{partI} will maintain a $(\beta,(1-\lambda')\beta)$-EDCS (so $\eps$ as used in \Cref{partI} is replaced by $\lambda'$),
but more importantly,  %, as in \Cref{partI}, within $O_\eps(\beta)$ worst-case update time.
by \Cref{crux}, the algorithm of \Cref{partII} maintains a $(\gamma,\gamma^-)$-EDCS, for
$\gamma^- = (1-2\lambda')\gamma = (1-\lambda)\gamma$. 
It remains to verify that these parameters satisfy the conditions of \Cref{keyl}, indeed:
$\lambda = 2 \lambda' = \frac{\eps}{100}$,  $\gamma > \beta \ge 32 \lambda^{-3}$ and $\gamma^- = (1-\lambda)\gamma$.
Thus $H$ provides a $(3/2+\eps)$-approximate matching sparsifier for $G$, as required.
% to satisfy the con\Cref{keyl}:
%Then we set the parameters $\eps$ and $\beta$, which were used for the $(\beta,(1-\eps))$-EDCS of \Cref{partI}, as follows: $\eps = %\frac{\eps'}{200}$, $\beta = 32 \eps^{-3}$.
 
%We take $\beta$ to be the maximum among $\sqrt{\Delta} \cdot \eps$ and $32 \lambda^{-3}$. 
By \Cref{tweakedl}, the worst-case update time of maintaining
$H$ is $O(\frac{\Delta}{\epsilon\beta} + 1/\eps)$. %  = O(\sqrt{\Delta} + 1/\eps)$. %\shay{this isn't the bottleneck anyway}
%By \Cref{recourse}, the worst-case recourse bound of maintaining $H$ is $O(1/\eps)$.
% which implies that the worst-case update time
%is
%\shay{I haven't read the following para carefully (I did see that it needs some polish), but it should be given after Part III. 
%Also, I suggest to structure it similarly to ``Proof of Theorem 1'' or ``Proof of Theorem 2'' in BS15; in particular, they define the notion of update ratio, and we may want to do the same. In general, we can copy-paste from them, as we did with Sec 3.}  
Moreover, note that $H$ has a maximum degree of $O(\beta)$ at all times. We use the Gupta-Peng \cite{GP13} algorithm, whose worst-case update time per change in $H$ is $O(\frac{\beta}{\epsilon^2})$,
to maintain a $(1+\epsilon)$-MCM on top of $H$. 
Combining the approximation ratio of the EDCS from Lemma~\ref{keyl} with the $(1+\epsilon)$ approximation of the Gupta-Peng \cite{GP13} algorithm yields a $(3/2+\varepsilon)(1+\varepsilon) \leq (3/2 + 7/2\varepsilon)$ approximation. Rescaling $\varepsilon$ by a factor $7/2$ yields the desired $(3/2+\eps)$-approximation, which will only have a constant factor effect on the update time.
By \Cref{rec}, the worst-case recourse bound of $H$ is $O(1/\eps)$,
%In the worst case, a single update in $G$ causes $\Theta(\frac{1}{\epsilon})$ updates in $H$ (number of alternating edges), 
thus a single update in $G$ is handled within a worst-case update time of $O(\frac{\beta}{\epsilon^3})$. 
It follows that the worst-case update time for maintaining a $(1+\eps)$-MCM on top of $H$ is 
$O(\frac{\Delta}{\epsilon\beta} +  \frac{\beta}{\epsilon^3})$, for any parameter $\beta \ge 32 \lambda^{-3}$.
We can optimize by substituting $\beta$ with $\eps \sqrt{\Delta}$ (assuming $\eps \sqrt{\Delta} \ge 32 \lambda^{-3}$, otherwise we need to increase $\beta$ accordingly), which yields:
%Summarizing, we have proved the following:
\begin{lemma}
\label{deg_optimized}
The worst-case update time for maintaining a $(1+\eps)$-MCM for $H$ is $O(\sqrt{\Delta}/\eps^2 + \eps^{-6})$.
%$O(\frac{\Delta}{\epsilon\beta} +  \frac{\beta}{\epsilon^3})$, for any parameter $\beta \ge 32 \lambda^{-3}$. 
\end{lemma}
%{\bf Remark.} Note that the $(1+\eps)$-MCM that we maintain for $H$ is actually a $(1+\eps)^2$-MCM for the entire graph $G$, but we can %reduce the approximation back to $1+\eps$ by a scaling argument; this will increase the update time by at most a constant factor.
%\chris{This is the worst case update time already combined with Gupta/Peng. My suggestion:}

%\begin{lemma}
%\label{deg_optimized}
%The deterministic worst-case update time for maintaining a $(3/2+\eps)$ approximate matching per edge insertion/deletion is 
%$O(\frac{\Delta}{\epsilon\beta} +  \frac{\beta}{\epsilon^3}) = O(\min(\varepsilon\sqrt{\Delta},\eps^{-6})$, for any parameter $\beta \ge 32 %\lambda^{-3}$. 
%\end{lemma}

 % in the worst case when applying GP'13 on $H$.

% Each update in $G$ takes another $O(\frac{\Delta}{\epsilon\beta}log(n))$ update time in the worst case (updating the subset of neighbors at %the end of the path, and in each update in $G$ we have only up to 2 of these vertices). Thus, in total, the update time is %$O(\frac{\Delta}{\epsilon\beta}log(n) + \frac{\beta}{\epsilon^3})$. The value for $\beta$ which optimizes this is% $\beta=\epsilon\sqrt{\Delta log(n)}$, and it gives a total update time of $O(\frac{\sqrt{\Delta log(n)}}{\epsilon ^2})=O_\epsilon(\sqrt{\Delta log(n)})$. 

%% file: PartIII_v2.tex
%\section{Improving the Worst Case Bound to $O_\epsilon(\sqrt{\Delta log(n)})$}
\section{Input Sparsification, Part III: Final Update Time $O_\eps(m^{1/4})$ or $O_\eps(\sqrt{\alpha})$}
\label{partIII}

In this section we show how to combine a simple {\em sparsification} step with the $O_\eps(\sqrt{\Delta})$ update time algorithm of 
\Cref{partII} (Part II) to achieve the desired update time of  $O_\eps(\sqrt{\alpha})$, which is  $O_\eps(m^{1/4})$ in general graphs.
%improve the update time in case of sparser graphs.

The  key component is a matching sparsification algorithm by Solomon~\cite{Sol18}, which works as follows. Suppose the arboricity of a given $m$-edge graph $G$ is bounded above by $\alpha$; as mentioned, $\alpha$ is no greater than $\sqrt{m}$. 
Every vertex marks up to $\eta$ arbitrary incident edges; that is, if a vertex has at most $\eta$ incident edges then it marks all of them,
otherwise it marks exactly $\eta$ arbtirary incident edges.
An edge that is marked twice (by both endpoints) is added to the sparsified graph, denoted by $G' = G'_\eta$. 
In \cite{Sol18} it was shown that for $\eta= 5(5/\eps + 1)2\alpha$, $G'$ is a $(1+\epsilon)$-approximate matching sparsifier for $G$,
i.e., $\mu(G) \le (1+\eps)\mu(G')$. Moreover, the maximum degree of $G'$ is trivially bounded by $\eta = O(\alpha/\eps)$. 

Thus, our goal is to dynamically maintain $G'$ on top of the input dynamic graph $G$, and feed $G'$ (rather than $G$) to the algorithm from Part II. In what follows we fix $\eps > 0$ to be an arbitrary parameter.

%\subsection
\paragraph{A density-sensitive degree bound.~} %\label{fixed}
First, we shall assume that $\alpha$ is a fixed upper bound on the arboricity of the dynamic graph $G$,
and show that the sparsifier $G'$ can be maintained 
with a constant worst-case update time and recourse bound, where the maximum degree of $G'$ is always  $O_\eps(\alpha)$.
%\shay{rephrase?}

\begin{lemma}
\label{lem:sparsearboricity}
%Let $\alpha$ be an upper bound on the arborticity. 
We can maintain $G' = G_\eta$, for $\eta:=5(5/\eps + 1)2\alpha$, with a constant worst-case (deterministic) update time and recourse bound.
%Moreover, the recourse bound (the number of changes to $G'$) is also constant in the worst case.
\end{lemma}
\begin{proof}
For every vertex $u$, we maintain a doubly-linked list of incident marked and unmarked edges, denoted by $LM(u)$ and $LU(u)$, respectively;
we also maintain a doubly-linked list of its incident edges in $G'$, denoted by $G'(u)$.
 Furthermore, every edge has mutual pointers to its various incarnations in all the lists it is contained in, as well as to the endpoints of the edge.

When a new edge $e=(u,v)$ is added to $G$, we add $e$ to $LM(u)$ (resp., $LM(v)$) if $|LM(u)|<\eta$ (resp., $|LM(v)| < \eta$), otherwise we add $e$ to $LU(u)$ (resp., $LU(v)$). If $e$ belongs to both $LM(u)$ and $LM(v)$, it is also added to $G'(u)$ and $G'(v)$, and thus to $G'$.
Clearly, the update time is constant and $G'$ can change by at most one edge.

When an edge $e = (u,v)$ is deleted from $G$, if $e\in LU(u)$ (resp. $LU(v)$), we simply remove it. 
If $e$ is in $G'$, we remove it from there by removing it from the corresponding lists $G'(u)$ and $G'(v)$.  %, and thus from $G'$.
If $e\in LM(u)$ (resp., $e\in LM(v)$), we remove an arbitrary edge from $LU(u)$ (resp. $LU(v)$) if one exists and insert it into $LM(u)$ (resp., $LM(v)$); 
in the latter case, if an edge that is moved to $LM(u)$ (resp., $LM(v)$) is also marked by the other endpoint, it is also added to
$G'$.
Here too the update time is constant, but now $G'$ can change by at most three edges: $e$ may be deleted from $G'$, and the edges moved from $LU(u)$ and $LU(v)$ to $LM(u)$ and $LM(v)$, respectively, if any, may be inserted to $G'$.
% , denoted $(u,x)$, may be inserted to $G'$ (if $(u,x) \in LM(x)$),  and (3) an edge moved from $LU(v)$ to $LM(v)$, denoted $(v,w)$ (if $(v,w) \in LM(w)$).
\qed
\end{proof}

\begin{corollary} \label{cor1}
One can dynamically maintain a $(1+\eps)$-approximate matching sparsifier $G' = G'_\eta$ for any dynamic graph $G$ with a constant worst-case update time and recourse bound, where the maximum degree of $G'$ is always upper bounded by $\eta= 5(5/\eps + 1)2\alpha$ and $\alpha$ is a fixed upper bound on the arboricity.
\end{corollary}

%\subsection
\paragraph{A degree bound of $O(\sqrt{m}/\eps)$, for a dynamic $m$.~}
%The drawback of \Cref{lem:sparsearboricity} is that the given arboricity bound $\alpha$ is fixed.
%implies that the sparsifier $G'$ can be maintained with a worst-case update time and recourse bound, and its degree is upper bounded by $O_{\eps}(m^{1/4})$. 
%In~\Cref{fixed} we assumed that $\alpha$ is a fixed upper bound on the arboricity of the graph, and managed to maintain $G'$ with a constant worst-case update time and recourse bound, where the maximum degree of $G'$ is always upper bounded by $O_\eps(\sqrt{\alpha})$. 
Next, we shall drop the assumption  that an upper bound $\alpha$ on the arboricity is given.
Starting from a graph with no edges, we will show how to maintain the sparsifier $G'$ with a constant worst-case update time and recourse, where the maximum degree of $G'$ is always  $O_{\eps}(\sqrt{m})$, where $m$ is the dynamically changing number of edges in the graph.

Clearly, in the degenerate case that the number of edges always remains the same up to, say, a factor of 2, we can use the argument of \Cref{lem:sparsearboricity} 
verbatim, by setting $\eta$ to be larger than $5(5/\eps + 1)2\sqrt{m}$ by a small constant factor.
The general case is handled using a standard trick of {\em periodic restarts}.
When the number of edges grows from $m$ to $2m$ or shrinks from $m$ to $m/2$, we shall restart the sparsifier $G'$ to be in accordance with the new value of the number of edges. For concreteness, if the number of edges at the moment of restart is denoted by $m_R$,
we set $\eta$ to be $5(5/\eps + 1) 4\sqrt{m_R}$, which is a factor $2$ larger than the bound used in $\Cref{lem:sparsearboricity}$. 
If we aimed for an amortized bound, a restart would be straightforward. Indeed, if the number of edges grows from $m$ to $m_R = 2m$, we set $\eta$ to $5(5/\eps + 1)4\sqrt{m_R} = 5(5/\eps + 1)4 \sqrt{2} \sqrt{m}$, and then scan, for each non-isolated vertex $v$ in the graph, all its incident edges, in order to mark its incident edges until (at most) $\eta$ are ''marked'', which involves moving edges from $LU(v)$ to $LM(v)$ until 
$|LM(v)| = \min\{\eta,|N(v)|\}$ and updating $G'(v)$ accordingly. The case that the number of edges decreases from $m$ to $m_R = m/2$ is handled symmetrically: we set $\eta$ to $5(5/\eps + 1)4\sqrt{m_R} = 5(5/\eps + 1)2\sqrt{2}\sqrt{m}$, and then scan for each non-isolated vertex $v$, all its incident edges, in order to ``unmark'' incident edges until (at most) $\eta$ are marked, which involves moving edges from $LM(v)$ to $LU(v)$ until $|LM(v)| = \min\{\eta,|N(v)|\}$ and updating $G'(v)$ accordingly.
In both cases this restart takes $O(m)$ time, assuming we keep track of   the non-isolated vertices in a separate list from the isolated ones, which can be amortized over all the edge updates occurred since the previous restart  to achieve a constant {\em amortized} update time and recourse. Moreover, at any point in time, since the values of $m_R$ and $m$ (the number of edges at the time of the last restart and at the current time, respectively) differ by at most a factor of 2,
the value of $\eta$, namely $5(5/\eps + 1)4\sqrt{m_R}$, is the same as $5(5/\eps + 1)4\sqrt{m}$ up to a factor of $\sqrt{2}$.
As $\eta$ upper bounds the maximum degree in $G'$, we conclude that the maximum degree of $G'$
is at most $5(5/\eps + 1)4 \sqrt{2} \sqrt{m}$.
Moreover, $\eta$ is at least $5(5/\eps + 1)2 \sqrt{2} \sqrt{m}$, which allows us to apply \cref{lem:sparsearboricity};
in fact,  we have an extra factor of $\sqrt{2}$ on the value $\eta$, which is not needed here, but will be needed next for achieving a worst-case bound. 
%between $5(5/\eps + 1)2\sqrt{2}\sqrt{m}$ and $5(5/\eps + 1)2\sqrt{2}\sqrt{m_R}$
 
Since we are interested in achieving a constant {\em worst-case} update time and recourse, we must implement the restart process {\em gradually} rather than instantaneously, simulating $c$ ``computational steps'' rather than 1 per each update step, for a   large constant $c$ of our choice, where each computational step corresponds to a sufficiently large constant running time. 
Of course, each edge insertion and deletion to and from the graph should also be processed, as it normally would, so we process it as in \Cref{lem:sparsearboricity} but according to the new value of $\eta$ that was set at the start of the last restart --- which absorbs 1 computational step out of $c$. Then we use the remaining $c-1$ computational steps to continue the restart process from where we left off in the previous edge update. 
Since we are simulating $c$ computational steps per edge update and as we can control the value of $c$, the entire restart process will be completed within less than $\rho \cdot m$ update steps, where $\rho \ll 1$ is a  small constant of our choice, at which stage the number of edges is between $m_R - \rho \cdot m$ and  $m_R + \rho \cdot m$. The worst-case update time and recourse bound are both constants by design, and it is easy to see that the maximum degree of the maintained sparsifier $G'$ is equal,
up a  constant factor, to  
$5(5/\eps + 1)4\sqrt{m}$, where $m$ is the dynamic number of edges.
The aforementioned extra factor of $\sqrt{2}$ on the value of $\eta$ allows us to apply \cref{lem:sparsearboricity} in this case too.
Summarizing, we have proved:

\begin{corollary} \label{cor2}
One can dynamically maintain a $(1+\eps)$-approximate matching sparsifier $G' = G'_\eta$ for any dynamic graph $G$ with a constant worst-case update time and recourse bound, where the maximum degree of $G'$ is within a constant factor of $\eta= 5(5/\eps + 1)4\sqrt{m}$
% and $\alpha$ is a fixed upper bound on the arboricity.
%Alternatively, one can achieve a degree bound of $\Theta(\sqrt{m}/\eps)$, where
and $m$ is the dynamic  number of edges.
\end{corollary}

\paragraph{Completing the proofs of \Cref{thm1} and \Cref{thm2}.~}
Combining \Cref{cor1} and \Cref{cor2} in conjunction  with \cref{deg_optimized}, we maintain in this way a $((3/2+\eps)\cdot (1+\eps))$-MCM for the original graph $G$, which provides a $(3/2+\eps)$-MCM after rescaling $\eps$ by  a constant factor.
%The overall running time is dominated by running the Gupta-Peng algorithm~\cite{GP13} algorithm on the EDCS with $\beta \in %O(\min(\sqrt{\alpha\cdot \eps},\eps^{-3}))$ for bounded arboricity graphs or $\beta \in O(\min(m^{1/4}\cdot \eps^{0.5},\eps^{-3})$ for general graphs. 
By  \Cref{cor1} and \Cref{cor2}, the worst-case recourse bound of maintaining $G'$ is constant, hence by \cref{deg_optimized} the overall worst-case update time
is $O(\sqrt{\Delta}/\eps^2 + \eps^{-6})$, which is $O(m^{1/4} \eps^{-2.5} + \eps^{-6})$ for general graphs
and  $O(\sqrt{\alpha} \eps^{-2.5} + \eps^{-6})$ for graphs with arboricity bounded by $\alpha$.
% at all times. This proves 
\Cref{thm1} and \Cref{thm2} follow.
%
%Thus, the overall update time is $O(\min(\sqrt{\alpha}\eps^{-2.5},\eps^{-6}))$ and $O(\min(m^{1/4}\eps^{-2.5},\eps^{-6}))$, respectively.